\documentclass[conference]{IEEEtran}
\IEEEoverridecommandlockouts
\ifCLASSINFOpdf
\else
\fi

\usepackage{tabularx}
\usepackage{threeparttable}
\usepackage{mathrsfs}
\usepackage{color}
\usepackage{amsfonts,amssymb}
\usepackage{amsfonts}
\usepackage{subfigure}
\usepackage{booktabs}

\usepackage{amsmath}
\usepackage{enumerate}
\usepackage{algorithm}
\usepackage{algorithmic}
\usepackage{bm}
\usepackage{makecell}
\usepackage{multirow}
\usepackage{mathtools}
\usepackage{bbm}
\usepackage{dsfont}
\usepackage{pifont}
\usepackage{amsthm}
\usepackage{cite}
\usepackage{balance}

\newtheorem{remark}{Remark}
\newtheorem{theorem}{Theorem}
\newtheorem{lemma}{Lemma}

\newtheorem{assumption}{Assumption}

\newtheorem{definition}{Definition}

\newtheorem{corallary}{Corallary}

\begin{document}

\title{
Inferring Topology of Networked Dynamical Systems by Active Excitations}
\author{Yushan Li, Jianping He, Cailian Chen and Xinping Guan
\thanks{The authors are with Dept. of Automation, Shanghai Jiao Tong University, Key Laboratory of System Control and Information Processing, Ministry of Education of China, and Shanghai Engineering Research Center of Intelligent Control and Management, Shanghai, Chin. E-mail address: \{yushan\_li, jphe, cailianchen, xpguan\}@sjtu.edu.cn. 
}%
}

\maketitle

\begin{abstract}
Topology inference for networked dynamical systems (NDSs) has received considerable attention in recent years. 
The majority of pioneering works have dealt with inferring the topology from abundant observations of NDSs, so as to approximate the real one asymptotically. 
Leveraging the characteristic that NDSs will react to various disturbances and the disturbance's influence will consistently spread, this paper focuses on inferring the topology by a few active excitations. 
The key challenge is to distinguish different influences of system noises and excitations from the exhibited state deviations, where the influences will decay with time and the exciatation cannot be arbitrarily large.   
To practice, we propose a one-shot excitation based inference method to infer $h$-hop neighbors of a node. 
The excitation conditions for accurate one-hop neighbor inference are first derived with probability guarantees. 
Then, we extend the results to $h$-hop neighbor inference and multiple excitations cases, providing the explicit relationships between the inference accuracy and excitation magnitude. 
Specifically, the excitation based inference method is not only suitable for scenarios where abundant observations are unavailable, but also can be leveraged as auxiliary means to improve the accuracy of existing methods. 
Simulations are conducted to verify the analytical results.


\end{abstract}

\section{Introduction}

Networked dynamical systems (NDSs) have been extensively used in numerous applications in the last decades, e.g., electric power systems \cite{cavraro2018graph}, transportation systems \cite{deri2016new}, and multi-robot systems \cite{lys21ecc}. 
The topology of NDSs is fundamental to characterizing interactions between individual nodes and determines the system convergence. 
Inferring the topology from observations provides insightful interpretability about NDSs and associated task implementations, and has become a hotspot research topic.  

In the literature, plenty of works have been developed to address the topology inference problem from different aspects \cite{brugere2018network}. 
For instance, in terms of static topology, \cite{etesami2017measuring,santos2019local,lys-cdc2021,jiao2021topologya} focus on inferring the causality/dependency relationships between nodes, while \cite{segarra2017network,segarra2017network-cdc,zhu2020network} reconstruct the topology by finding the most suitable eigenvalues and eigenvectors from the sample covariance matrix. 
Considering the topology is time-varying by rules, available methods include graphical Lasso-based methods \cite{gibberd2014high} and SEM models \cite{baingana2016tracking}, which take the varying topology as a sequence of static topologies and infer them, respectively. 
In addition to the dynamic topology inference, many kernel-based methods are proposed to deal with cases with nonlinear system models \cite{karanikolas2016multi,wang2018inferring}. 

Despite the tremendous advances of the above works, almost all of the approaches are based on a large scale of observations over the systems. 
In other words, the feasibility lies in digging up the regularity of the dynamical evolution process from the observation sequences, which corresponds to the common intuition that more data make the interpretability better \cite{dong2019learning}. 
Unfortunately, when the observations over the NDS are very limited, the aforementioned methods cannot work well. 
For example, for a linear time-invariant NDS of $n$ nodes, at least $(n+1)$ groups of consecutive global observations are required to obtain a unique least square estimate of the topology matrix. 
When more observations are not allowable due to some practical limitations, directly inferring the topology from observations will be extremely difficult. 

Inspired by the phenomenon that a thrown stone into water will cause waves, we are able to proactively inject inputs into the systems to excite corresponding reaction behaviors, i.e., the injected inputs on one node will spread to other neighbor nodes.  
Related examples include using Traceroute to probe the routing topology of the Internet \cite{holbert2015network}, or utilizing inverters to probe the electric distribution network \cite{cavraro2019inverter}. 
Therefore, it is possible to reveal the underlying topology of NDSs by investigating the relationships between the excitations and reactions \cite{weerts2015identifiability,hendrickx2018identifiability,bazanella2019network}. 
This idea has motivated the study of this paper, where we aim to leverage a few active excitations to do the inference tasks. 
It is worth noting that if the excitations are allowed to be abundant, then the problem falls into the realm of typical system identification \cite{coutino2020state,xyminput2022acc}, which is not the focus of this paper. 

Few excitations indicate small inference costs but incur new challenges. 
On the one hand, the influence of the excitation is closely coupled with that of stochastic noises, making it hard to directly distinguish their difference. 
On the other hand, the spreading effect of the excitation will decay with time and the excitation cannot be arbitrarily large, limiting the scope and accuracy of the inferred topology. 
To address these issues, we introduce the probability measurement to infer the topology from a local node, and demonstrate how to determine whether the information flow between two nodes exists. 
The main contributions are summarized as follows.

\begin{itemize}
\item We investigate the possibility of inferring the topology of NDSs by a few active excitations, taking both the process and measurement noises into account. 
Specifically, we utilize hypothesis test to establish criteria of how to determine the connections between nodes from the exhibited state deviations after excitations.  

\item Considering the spreading effects of excitations in NDSs, we first propose one-shot excitation based method to infer one-hop neighbors of a single node. 
Then, we prove the critical excitation condition given tolerable misjudgment probability, providing reliable excitation design guidance. 

\item Based on the one-hop inference procedures, we extend the theoretical analysis to multi-hop neighbor inference by one-shot excitation and multiple excitation cases, respectively. 
The relationship between inference accuracy and excitation magnitude is derived with probability guarantees.  
Simulations verify our theoretical results. 
\end{itemize}

The proposed inference method by a few active excitations applies to situations where the observations about NDSs are not sufficient. 
It can also be leveraged as an auxiliary measure to enhance the accuracy of existing methods by large scales of observations, by treating the inferred results as the constraints in counterpart problem modeling. 
The remainder of this paper is organized as follows.
In Section \ref{sec:preliminaries}, some preliminaries of NDSs and problem modeling are presented. 
The inference method and performance analysis are provided in Section \ref{sec:results}.
Simulation results are shown in Section \ref{sec:simulation}.
Finally, Section \ref{sec:conclusion} concludes the paper.

\section{Preliminaries and Problem Formulation}\label{sec:preliminaries}
\subsection{Graph Basics and Notations}
Let $\mathcal{G}=(\mathcal{V},\mathcal{E})$ be a directed graph that models the networked system, where $\mathcal{V}=\{1, \cdots, n\}$ is the finite set of nodes and $\mathcal{E}\subseteq \mathcal{V}\times \mathcal{V}$ is the set of interaction edges.
An edge $(i,j)\in \mathcal{E}$ indicates that $i$ will use information from $j$.
The adjacency matrix $A=[a_{ij}]_{n \times n}$ of $\mathcal{G}$ is defined such that ${a}_{ij}\!>\!0$ if $(i,j)$ exists, and ${a}_{ij}\!=\!0$ otherwise.
Denote ${\mathcal{N}_i}=\{j\in \mathcal{V}:a_{ij}>0\}$ as the in-neighbor set of $i$, and $d_i=\left| {\mathcal{N}_i} \right |$ as its in-degree. 
Throughout this paper, let $\bm{0}$ and $\bm{1}$ be all-zero and all-one matrices in compatible dimensions, $\rho_{\min}(M)$ and $\rho_{\max}(M)$ be the smallest and largest eigenvalues of the matrix $M$, respectively.

\subsection{System Model}
Consider the following networked dynamical model
\begin{equation}\label{eq:global_model}
\begin{aligned}
x_{t}&=Wx_{t-1}+\theta_{t-1}, \\
y_{t}&=x_{t}+\upsilon_{t},
\end{aligned}
\end{equation}
where $x_t$ and $y_t$ represents the system state and corresponding observation at time $t$, $W=[w_{ij}]_{n \times n}$ is the interaction topology matrix related to the adjacent matrix $A$, 
and $\theta_{t}$ and $\upsilon$ represent the process and observation noises, satisfying the following Gauss-Markov assumption. 
\begin{assumption}
$\theta_{t}$ and $\upsilon_{t}$ are i.i.d. Gaussian noises, subject to ${N}(0,\sigma^2_{\theta} I)$ and ${N}(0,\sigma^2_{\upsilon} I)$, respectively. 
They are also independent of $\{x_{t'}\}_{t'=0}^{t'=t}$ and $\{y_{t'}\}_{t'=0}^{t'=t}$. 
\end{assumption}

Next, we characterize the stability of \eqref{eq:global_model} by defining 
\begin{equation}
\begin{aligned}
\mathcal{S}_a=&\{Z\in \mathbb{R}^{n \times n}, \rho_{\max}(Z)<1\}, \\
\mathcal{S}_m=&\{Z\in \mathbb{R}^{n \times n}, \rho_{\max}(Z)=1~\text{~and the geometric} \\
&\text{multiplicity of eigenvalue 1 equals to one} \} .
\end{aligned}
\end{equation}
Then, $W$ is called asymptotically stable if $W\in\mathcal{S}_a$, or marginally stable matrix if $W\in\mathcal{S}_m$. 
Concerning its setup, popular choices include the Laplacian and Metropolis rules \cite{sayed2014adaptation}, which are given by 
\begin{align}
w_{ij}&=\left \{
\begin{aligned}
&\gamma a_{ij}/ \max\{d_i, i\in\mathcal{V}\},&&\text{by Laplacian rule} ,\\
&a_{ij}/ \max \{d_{i}, d_{j}\},&&\text{by Metropolis rule},
\end{aligned}
\right. \label{eq:topo-rule} \\
w_{ii}&= 1-\sum\nolimits_{j \neq i} w_{ij}. \label{eq:topo-rule2}
\end{align}
where the auxiliary parameter $\gamma$ satisfies $0<\gamma\le 1$. 
Note that if $W$ is specified by either one of the two rules, then $W\in\mathcal{S}_m$.
A typical matrix in $\mathcal{S}_a$ can be directly obtained via multiplying (\ref{eq:topo-rule}) and (\ref{eq:topo-rule2}) by a factor $0<\alpha<1$, which is common in adaptive diffusion networks \cite{matta2018consistent}. 
Based on \eqref{eq:global_model}, the observation $y_t$ can be recursively expanded as
\begin{align} \label{eq:expand_form}
y_{t}=x_t+\upsilon_t=W^{t}x_{0}+\sum\nolimits_{m = 1}^{t} W^{m-1} \theta_{t-m} + \upsilon_t.
\end{align}
Considering different stabilities, it holds that
\begin{equation}
\mathop {\lim }\limits_{t \to \infty } W^t=\left \{
\begin{aligned}
&\bm{0},~&&\text{if}~W \in \mathcal{S}_a \\
&W^{\infty},~&& \text{if}~W \in \mathcal{S}_m,
\end{aligned}\right.
\end{equation}
where $\|W^{\infty}\|<\infty$. 
Therefore, if $W\in \mathcal{S}_a \cup \mathcal{S}_m$, $y_t$ is strictly bounded in the expectation sense.

\subsection{Inference Modeling and Problem of Interest}
Since the information flow between nodes is specified by the topology of NDSs, we first define the $h$-hop neighbor of a single node. 
\begin{definition}[$h$-hop out-neighbor]
Node $i$ is a $h$-hop out-neighbor of node $j$ if the minimal edge number of an acyclic path from $j$ to $i$ is $h$, satisfying
\begin{equation}\label{}
\prod\limits_{l = 1}^{h} {{a_{{i_l}{i_{l + 1}}}}}  = {a_{{i_1}{i_2}}}{a_{{i_2}{i_3}}}{a_{{i_3}{i_4}}}...{a_{{i_{h - 1}}{i_{h}}}}{a_{{i_{h }}{j}}}>0,
\end{equation}
where node ${i_1}=i$ and ${i_{h+1}}=j$. 
All the $h$-hop out-neighbors are represented by the set $\mathcal{N}_{j,h}^{out}$. 
\end{definition}
Note that when the topology is undirected, there is no need to differentiate the in/out-neighbors. 
If $j\in\mathcal{N}_{j,h}^{out}$, node $j$ is also called the $h$-hop in-neighbor of node $i$. 
Unless otherwise specified, we mainly focus on the $h$-hop out-neighbors of a node in the following. 
To present an explicit expression for $\mathcal{N}_{j,h}^{out}$, we first define $\mathcal{N}_{j,h}^{e}$ as the node set where all the nodes can reached from node $j$ within $h$ hops. 
Then, $\mathcal{N}_{j,h}^{out}$ is recursively formulated as  
\begin{equation}\label{eq:hop_neighbor}
\mathcal{N}_{j,h}^{out}=\mathcal{N}_{j,h}^{e}\backslash \left \{ \mathop \cup \limits_{ l=1 }^{h-1} \mathcal{N}_{j,l}^{out} \right \}. 
\end{equation}
When $h=1$, $\mathcal{N}_{j,1}^{out}= \mathcal{N}_{j,1}^{e}$. 
The following assumption is made throughout this paper. 
\begin{assumption}\label{as:stable}
The topology matrix $W \in  \mathcal{S}_a \cup\mathcal{S}_m$, and the elements of $W$ are all non-negative. 
For all $w_{ij}>0$, there exists a lower bound $\underline{w}$ such that $w_{ij}\ge\underline{w}>0$.  
\end{assumption}

Finally, the problem of interest is formulated as follows. 
Consider that there are no sufficient observations of the NDS model \eqref{eq:global_model} to support existing estimation or regression methods of inferring topology, e.g., the causality based estimator in \cite{lys-cdc2021}. 
Leveraging the characteristic that a NDS is easily subjected to various disturbances and exhibits state deviation, 
we aim to reduce the inference dependence on observation scales, and propose an 
active excitation based method to infer the topology from limited new observations. 
Mathematically, let $e_{t}^{j}$ be the excitation input on node $j$ at time $t$, and $\tilde y_{t + h}^{i}$ be the observation of $i$ at time $(t+h)$. 
Then, the goal of this paper is to find $\mathcal{N}_{j,h}^{out}$ from the limited observations $\{y_{t},\tilde y_{t + h'},h'=1,\cdots,h \}$. 

The above problem is very challenging, as the stochastic process and measurement noises will also affect $\{\tilde y_{t + h'},h'=1,\cdots,h \}$ and even accumulate. 
We will address these issues from the following aspects. 
\begin{itemize}
\item To make the excitation's impact differentiable, we resort to the tool of hypothesis testing to derive the conditions for excitation magnitude to guarantee arbitrary misjudgement probability of one-hop neighbor inference. 
\item To overcome the decaying effects of excitations, we establish the probabilistic relationship between the one-shot excitation magnitude and the accuracy of $h$-hop neighbor inference. 
\item We further extend the analysis to multiple excitations cases and illustrate how to use excitations to improve the accuracy of existing topology inference methods. 
\end{itemize}


\section{Excitation-based Inference Method}\label{sec:results}
In this section, we first analyze the reaction behavior of a NDS under excitation inputs. 
Then, we focus on how to infer the one-hop neighbors by one-shot excitation and characterize the inference accuracy in probability. 
Finally, we discuss how to infer the $h$-hop neighbors and multiple excitations cases. 

\subsection{Observation Modeling Under excitation}

Since only $\{y_t\}_{t=0}^{T}$ are directly available, for every two adjacent observations, it follows that
\begin{align}\label{eq:two-observation}
y_t&=W (y_{t-1}-\upsilon_{t-1})+\theta_{t-1}+\upsilon_t \nonumber \\
&=Wy_{t-1}+\omega_t,
\end{align}
where $\omega_t=-W\upsilon_{t-1}+\theta_{t-1}+\upsilon_t$, satisfying $N(0,\sigma_{\upsilon}^2 WW^\mathsf{T} + \sigma_{\upsilon}^2 I + \sigma_{\theta}^2 I)$. 
Besides, $\omega_t$ is independent of all $\{x_{t'}\}_ {t'< t}$ and $\{\theta_{t'}\}_ {t'< t-1}$.
We point out that (\ref{eq:two-observation}) only represents the quantitative relationship between adjacent observations, not a causal dynamical process.

Similar to \eqref{eq:two-observation}, the observation at time $t+h$ can be recursively written as 
\begin{align}
y_{t+h}\!=\!\Gamma(h) y_{t}\!+\!v_{t+h}\!-\!\Gamma(h) v_{t}\!+\!\sum_{m=1}^{h} \Gamma(m\!-\!1) \theta_{t+h-m},
\end{align}
where $\Gamma(h)=W^{h}$ is the $h$-step translation matrix. 
For ease notation, let $\omega_{t,h}=v_{t+h}-\Gamma(h) v_{t}+\sum_{m=1}^{h} \Gamma(m-1) \theta_{t+h-m}$. 
Then, the deviation between $y_{t+h}^{i}$ and $y_{t}^{i}$ is represented by
\begin{align}\label{eq:deviation_h}
y_{t+h}^{i}-y_{t}^{i}&=[\Gamma(h)y_{t}]^{i}-y_{t}^{i}+\omega_{t,h}^{i},
\end{align}
where $\omega_{t,h}^{i} \!\sim\! N(0, \sigma_{\omega,h}^2(i) )$ and $\sigma_{\omega,h}^2(i)$ is given by 
\begin{equation}
\sigma_{\omega,h}^2(i) \!=\! \left(1\!+\!\sum\limits_{j = 1}^{n}\Gamma_{ij}^2(h) \right)\sigma_{\upsilon}^2 \!+\! \left(\sum\limits_{m= 1}^{h} \sum\limits_{j = 1}^{n}\Gamma_{ij}^2(m\!-\!1)\right) \sigma_{\theta}^2,
\end{equation}
which is obtained from the mutual independence of the process and observation noises. 

Note under Assumption \ref{as:stable}, it holds that $\Gamma(h)\bm{1}\!\le\!\bm{1}$ and $\sum\nolimits_{j = 1}^{n}\Gamma_{ij}^2(h)\!\le\!1$. 
Leveraging the two properties, one can induce that 
\begin{align}
|y_{t+h}^{i}-y_{t}^{i}| \le \Delta y_{t}^{\max},~\sigma_{\omega,h}^2(i) \le 2\sigma_{\upsilon}^2 +h \sigma_{\theta}^2,
\end{align}
where the deviation bound $\Delta y_{t}^{\max} $ is given by 
\begin{equation}
\Delta y_{t}^{\max} = \left\{
\begin{aligned}
&\max\{ |y_{t}^{i}-y_{t}^{j}|: i,j\in \mathcal{V}\}, &&\text{if}~W \in \mathcal{S}_m\\
&\max\{ |y_{t}^{i}|: i\in \mathcal{V}\} , &&\text{if}~W \in \mathcal{S}_a. 
\end{aligned}\right.
\end{equation}
It is worth noting that $\Delta y_{t}^{\max}$ will fluctuate around zero as $t$ increases in either case of $\mathcal{S}_m$ and $\mathcal{S}_a$. 
For simplicity without loss of generality, consider node $j$ is injected with positive excitation input $e_{t}^{j}>0$ at time $t$. 
Then, the observation deviation is given by 
\begin{equation}\label{eq:excitation}
\tilde y_{t,h}^{i,\Delta} \!=\!\tilde y_{t + h}^{i}\!-\!y_{t}^{i}\! \le \! \left\{
\begin{aligned}
&\Delta y_{t}^{\max} \!+\! \Gamma_{ij}e_{t}^{j}\!+\! \omega_{t,h}^{i},~\text{if}~\Gamma_{ij}>0,\\
&\Delta y_{t}^{\max} \!+ \!\omega_{t,h}^{i},~\text{if}~\Gamma_{ij}=0.
\end{aligned}
\right.
\end{equation}
Note that the term $\Gamma_{ij}e_{t}^{j}$ in (\ref{eq:excitation}) represents the influence of the excitation input $e_j$ over $i$ after $h$ steps. 
Hereafter, we will drop the subscript $t$ in the variables if it does not cause confusion. 

\begin{remark}
The excitation put $e_t^j$ cannot be arbitrarily large due to the internal constraints in NDSs, otherwise one can easily infer the connections by a extremely large excitation input, which makes the inference trivial. 
Therefore, it is of greater necessity to investigate the relationships between the inference accuracy and excitation magnitude, providing available excitation guidance to obtain accurate inference results with probability guarantees. 
\end{remark}

\subsection{One-hop Neighbor Inference}
After node $j$ is injected with excitation input $e_{t}^{j}$, the one-step observation deviation of node $i$ is given by 
\begin{equation}\label{eq:excitation_1}
\tilde y_{t + 1}^{i,\Delta} \!=\!\tilde y_{t + 1}^{i}\!-\!y_{t}^{i}\! \le \! \left\{
\begin{aligned}
&\Delta y_{t}^{\max} \!+\! w_{ij}e_{t}^{j}\!+\! \omega_{t+1}^{i},~\text{if}~w_{ij}>0,\\
&\Delta y_{t}^{\max} \!+ \!\omega_{t+1}^{i},~\text{if}~w_{ij}=0.
\end{aligned}
\right.
\end{equation}
For legibility, we temporarily assume $\Delta y_{t}^{\max}=0$ and take its influence into consideration after the analysis. 
Since $\tilde y_{t,1}^{i,\Delta}$ is associated with the stochastic process and measurement noises, finding $\mathcal{N}_{j}^{out}$ can be modeled as a typical binary hypothesis testing. 
The null and alternative hypothesis are respectively defined as
\begin{equation}\label{eq:hypothesis}
\left\{
\begin{aligned}
&H_0: i \notin \mathcal{N}_{j}^{out},\\
&H_1: i \in \mathcal{N}_{j}^{out}.
\end{aligned}
\right.
\end{equation}
Then, denote $\Pr\{{H_0}| \tilde y_{t,1}^{i,\Delta} \} $ ($\Pr\{{H_1}| \tilde y_{t,1}^{i,\Delta} \} $) as the probability that $H_0$ ($H_1$) holds given the observation $\tilde y_{t,1}^{i,\Delta}$. 
Then, we have the following decision criterion 
\begin{equation}\label{eq:hypothesis-decision}
\left\{
\begin{aligned}
&\Pr\{ {H_1}| \tilde y_{t,1}^{i,\Delta} \} \ge \Pr\{ {H_0}| \tilde y_{t,1}^{i,\Delta}\} ~ \Rightarrow ~ H_1 ~\text{holds},\\
&\Pr\{ {H_1}| \tilde y_{t,1}^{i,\Delta} \} < \Pr\{{H_0}| \tilde y_{t,1}^{i,\Delta} \} ~ \Rightarrow ~ H_0 ~\text{holds}, 
\end{aligned}
\right.
\end{equation}
which is also called the maximum posterior probability criterion. 
However, it is possible that (\ref{eq:hypothesis-decision}) is misjudged in the test, for example, $H_0$ is true but $H_1$ is decided (Type I Error) or $H_1$ is true but $H_0$ is decided (Type II Error). 
Accordingly, let $\Pr\{D_1|H_0\}$ be the false alarm probability and $\Pr\{D_0|H_1\}$ be the missed detection probability, respectively. 
Therefore, the overall misjudgement probability is given by 
\begin{equation}\label{eq:mis_error}
\delta_e=\Pr\{D_1|H_0\}+\Pr\{D_0|H_1\}. 
\end{equation}

Suppose the inference center has no prior information about $H_1$ and $H_0$, i.e., $\Pr(H_1)=\Pr(H_0)=0.5$. 
Under hypothesis testing (\ref{eq:hypothesis-decision}), the following result presents the probabilistic relationship between the inference accuracy and the injected excitation magnitude. 
\begin{theorem}[Critical excitation for one-hop neighbors]\label{th:hypothesis}
To ensure the misjudgement probability within a threshold $\bar \delta_e$, the excitation $e_{j}$ should satisfy 
\begin{equation}\label{eq:lower_bound}
|e^j|\ge \frac{2\sqrt{2} \sigma_{\omega({i}) }}{w_{ij}} \text{erf}^{-1}(1-\bar \delta_e),
\end{equation}
where the Gaussian error $\text{erf}(z)= \frac{ 2 }{\sqrt{\pi} } \int_{0}^{z}  \exp{ (-r^2) } \mathrm{d}r $ and $\text{erf}^{-1}(\cdot)$ is the reverse mapping of $\text{erf}(z)$. 
\end{theorem}
\begin{proof}
The proof consists of two steps. 
First, we prove the decision threshold $z_0$ is given by $z_0=\frac{w_{ij}e^j}{2}$. 
Then, we demonstrate the critical excitation magnitude under the $z_0$. 

For simplicity without losing generality, we begin with the case where the excitation input $e^j>0$. 
Note that $\omega^i$ is a continuous random variable, the likelihood ratio $l_r(z)$ in the test is given by
\begin{equation}
l_r(z)=\frac{ f_{\omega}(z|H_1) }{ f_{\omega}(z|H_0) },
\end{equation} 
where $f_{\omega}(\cdot)$ is the probability density function of $\omega^i$. 
Due to the prior probabilities $\Pr(H_1)=\Pr(H_0)$, the decision threshold $z_0$ satisfies
\begin{equation}\label{eq:likelihood}
l_r(z_0)=\frac{ f_{\omega}(z_0|H_1) }{ f_{\omega}(z_0|H_0) }= \frac{ \Pr\{H_1\} }{\Pr\{H_0\}}=1. 
\end{equation} 
Since $w^i \sim N(0,\sigma^2_{\omega})$, substituting $f_{\omega}(y)= \frac{ 1 }{\sqrt{2\pi}\sigma_{\omega} }\exp{ (-\frac{ z^2 }{2\sigma^2_{\omega} }) }$ into (\ref{eq:likelihood}), it yields that
\begin{equation}\label{eq:likelihood2}
l_r(z_0)=\frac{ \exp{ (-\frac{(z_0-w_{ij}e^j)^2 }{2\sigma^2_{\omega} }) }  } { \exp{ (-\frac{z_0^2 }{2\sigma^2_{\omega} }) } }=1.
\end{equation} 
It follows from (\ref{eq:likelihood2}) that $z_0^2-(z_0-w_{ij}e^j)^2=0 $, leading to  
\begin{equation}\label{eq:decision_threshold}
z_0=\frac{w_{ij}e^j}{2}. 
\end{equation}

Next, by the definition of $\delta_e$, one has  
\begin{small}
\begin{align}\label{eq:proof_error}
\delta_e=&\Pr\{D_1|H_0\}+\Pr\{D_0|H_1\} = \int_{z_0}^{+ \infty}\!\!\!\!\! \frac{ 1 }{\sqrt{2\pi}\sigma_{\omega} } \exp{ (-\frac{z^2 }{2\sigma^2_{\omega} }) } \mathrm{d}z  \nonumber \\
& +  \int_{- \infty}^{z_0} \! \frac{ 1 }{\sqrt{2\pi}\sigma_{\omega} } \exp{ (-\frac{ (z-w_{ij}e^j)^2 }{2\sigma^2_{\omega} }) } \mathrm{d}z .
\end{align}
\end{small} 
\!\!Substitute $z=z'+\frac{w_{ij}e^j}{2}=z'+z_0$ into (\ref{eq:proof_error}), yielding
\begin{small}
\begin{align}\label{eq:proof_error2}
\delta_e=&\Pr\{D_1|H_0\}+\Pr\{D_0|H_1\} \nonumber \\
=& \int_{0}^{+ \infty}\!\!\!\!\! \frac{ 1 }{\sqrt{2\pi}\sigma_{\omega} } \exp{ (-\frac{(z'+z_0)^2 }{2\sigma^2_{\omega} }) } \mathrm{d}z'  \nonumber \\
&+ \int_{- \infty}^{0} \! \frac{ 1 }{\sqrt{2\pi}\sigma_{\omega} } \exp{ (-\frac{ (z'-z_0)^2 }{2\sigma^2_{\omega} }) } \mathrm{d}z' \nonumber \\
=& 2 \int_{0}^{+ \infty}\!\!\!\!\! \frac{ 1 }{\sqrt{2\pi}\sigma_{\omega} } \exp{ (-\frac{(z'+z_0)^2 }{2\sigma^2_{\omega} }) } \mathrm{d}z' \nonumber \\
=& 2 \int_{z_0}^{+ \infty}\!\!\!\!\! \frac{ 1 }{\sqrt{2\pi}\sigma_{\omega} } \exp{ (-\frac{z^2 }{2\sigma^2_{\omega} }) } \mathrm{d}z.
\end{align}  
\end{small}
\!\!Note that $\int_{z_0}^{+ \infty}\!\!\!\!\! \frac{ 1 }{\sqrt{2\pi}\sigma_{\omega}} \exp{ (-\frac{z^2 }{2\sigma^2_{\omega} }) } \mathrm{d}z=(1-\text{erf}(\frac{z_0}{\sqrt{2}\sigma_{\omega}}))/2$, thus it yields that 
\begin{equation}\label{eq:final_delta}
\delta_e=1-\text{erf}(\frac{z_0}{\sqrt{2}\sigma_{\omega}}).
\end{equation}
Substituting $z_0=\frac{w_{ij}e^j}{2}$ and $\delta_e=\bar \delta_e$ into (\ref{eq:final_delta}), we obtain
\begin{equation}
e^j=\frac{2\sqrt{2}\sigma_{\omega}}{w_{ij}} \text{erf}^{-1}(1-\bar \delta_e).
\end{equation}
The result is likewise when $e^j<0$ due to the symmetry of Gaussian distribution. 
By the monotone increasing property of $\text{erf}(z)$, to guarantee $\delta_e\le \bar \delta_e$, the excitation input must satisfy $|e^j|\ge \frac{2\sqrt{2}\sigma_{\omega}}{w_{ij}} \text{erf}^{-1}(1-\bar \delta_e)$. 
The proof is completed. 
\end{proof}

Theorem \ref{th:hypothesis} gives the lower magnitude bound of the excitation input to guarantee the specified misjudgment probability in a single time. 
Given the excitation input $e^j$ satisfying (\ref{eq:lower_bound}), one has with probability at least $(1-\bar \delta_e)$ to accurately discriminate whether $i\in\mathcal{N}_j^{out}$. 
Note that the interaction weight $w_{ij}$ is not priorly known in reality. Thus the decision threshold in theory, $\frac{w_{ij}e^j}{2}$, is unavailable. 
However, we can enable the hypothesis test by specifying the least interaction weight that one wishes to discriminate between two nodes. 

To practice, since $\|W\|_{F}\le\sqrt{n}\|W\|\le\sqrt{n}$, thus we have $\sum\limits_{j = 1}^{n}w_{ij}^2 \le n$ and 
\begin{equation}
\sigma_{\omega}^2(i)\le(1+n)\sigma_{\upsilon}^2+\sigma_{\theta}^2=\bar \sigma_{\omega}^2. 
\end{equation}
Specifically, if $W$ is row-stochastic, then the upper bound $\bar \sigma_{\omega}^2$ can be further reduced to $2\sigma_{\upsilon}^2+\sigma_{\theta}^2$. 
Next, suppose that one aims to judge whether $i\in\mathcal{N}_j^{out}$ such that $w_{ij}>\underline{w}$, where $\underline{w}$ is the weight lower bound. 
Given the desired error probability bound $\bar\delta_e$ and the excitation input $e^j$ such that $e^j= \frac{2\sqrt{2} \bar\sigma_{\omega }}{ \underline{w} } \text{erf}^{-1}(1-\bar \delta_e)$, then with probability at $1-\bar \delta_e$ one can discriminate whether $i\in\mathcal{N}_j^{out}$ by 
\begin{equation}\label{eq:1-hop-decision}
\left\{
\begin{aligned}
&i\in\mathcal{N}_j^{out},&&~\text{if}~ |\tilde y^{i,\Delta}_{t+1}| \ge \Delta y_{t}^{\max} + \frac{\underline{w} |e^j|}{2}, \\
&i\notin\mathcal{N}_j^{out},&& 
~\text{else}, 
\end{aligned}
\right.
\end{equation}
where the parameters in (\ref{eq:1-hop-decision}) are all computable or known. 
Applying (\ref{eq:1-hop-decision}) to all other node and one can obtain an estimated set of $\mathcal{N}_j^{out}$. 
Note that although the observation deviation $| (Wy_{t})^{i}-y_{t}^{i} |$ will affect the performance of excitation based topology inference, its influence is strictly bounded under Assumption \ref{as:stable}. 
The whole procedures are summarized in Algorithm \ref{excitation-algo}, where we use the lower bound $\underline{w}$ that is sufficient to guarantee the accuracy probability. 


A direct result from Theorem \ref{th:hypothesis} is $\mathop {\lim }\limits_{|e^j| \to \infty } \delta_e=0$, which corresponds to the common intuition. 
As long as the excitation input is large enough, the one-hop neighbors of $j$ can always be inferred. 
Under this situation, it is also very likely that the two-hop (even more) out-neighbors of the excited node can also be identified by just single excitation. 

\begin{algorithm}[t]
    \caption{Excitation-based Topology Inference}
    \label{excitation-algo}
    \begin{algorithmic}[1]
    \REQUIRE{Observations $y_t$, target excited node $j$, desired lower bound of interaction weight $\underline{w}$, upper bound $\bar\sigma_{\omega}$, and tolerant error probability $\bar \delta_e$. }
    \ENSURE{Estimation of the one-hop out-neighbor of $j$, $\hat{\mathcal{N}}_{j}^{out}$.}
    \STATE Initialize $\hat{\mathcal{N}}_{j}^{out}=\emptyset$.
    \STATE Calculate the critical excitation $e^j= \frac{2\sqrt{2} \bar\sigma_{\omega }}{ \underline{w} } \text{erf}^{-1}(1-\bar \delta_e)$.
    \STATE Excite node $j$ with $e^j$ and obtain the observation $y_{t+1}$. 
    \STATE Compute $\Delta y_{t}^{\max} = \max\{ |y_{t}^{i}-y_{t}^{j}|: i,j\in \mathcal{V}\}$.
    \FOR {$i \in\mathcal{V}$}
    {
      \STATE Compute the observation deviation $ \tilde y^{i,\Delta}_{t,1} = y_{t+1}^{i}-y_{t}^{i}$.
      	\IF {$ \tilde y^{i,\Delta}_{t,1} > \Delta y_{t}^{\max} + \frac{\underline{w} e^j}{2}$}
		{
			\STATE $\hat{\mathcal{N}}_{j}^{out}= \hat{\mathcal{N}}_{j}^{out}\cup \{i\}$.
		}
	  	\ENDIF
    }
    \ENDFOR
    \RETURN The one-hop neighbor set estimation $\hat{\mathcal{N}}_{j}^{out}$.  
    \end{algorithmic}
\end{algorithm}

\subsection{Multi-hop Neighbor Inference}
In this part, we will demonstrate how to identify multi-hop out-neighbors of a node by single excitation. 
Similar with the hypothesis (\ref{eq:hypothesis}), we first define the following hypothesis that tests whether $i \in \mathcal{N}_{j,h}^{e}$, i.e., 
\begin{equation}\label{eq:hypothesis2}
\left\{
\begin{aligned}
&H_0(h): i \notin \mathcal{N}_{j,h}^{e},\\
&H_1(h): i \in \mathcal{N}_{j,h}^{e}.
\end{aligned}
\right.
\end{equation}
Note that \eqref{eq:hypothesis2} is a test using the observation deviation $\tilde y_{t,h}^{i,\Delta}$ to judge whether node $i$ is an out-neighbor of node $i$ within $h$ steps. 
Although it cannot infer the $h$-hop neighbor directly, valuable information can still be extracted for the final inference. 
To begin with, we present the following result.


\begin{lemma}[Critical excitation for neighbors within $h$ hops]\label{lemma:hypothesis2}
Under hypothesis test (\ref{eq:hypothesis2}), to ensure the misjudgement probability for all the neighbor within h-hop is lower than $\bar \delta_e$, the excitation $e_{j}$ should satisfy 
\begin{equation}
|e^j|\ge \frac{2\sqrt{2}\sigma_{\omega,h}}{ \Gamma_{ij}(h) } \text{erf}^{-1}(1-\bar \delta_e). 
\end{equation}
\end{lemma} 
\begin{proof}
Directly focusing on the $h$-step node response after the excitation input is injected on $j$, $\Gamma(h)$ becomes the equivalent topology that corresponds to the $h$-step process. 
Based on Theorem \ref{th:hypothesis}, when $|e^j|\ge \frac{2\sqrt{2}\sigma_{\omega,h}}{\Gamma_{ij}} \text{erf}^{-1}(1-\bar \delta_e)$ ensures the misjudgement probability is no more than $(1-\bar \delta_e)$, which completes the proof. 
\end{proof}

Note that Lemma \ref{lemma:hypothesis2} only illustrates how to reduce the misjudgement probability of $i \in \mathcal{N}_{j,h}^{e}$, and does not provide information about whether $i \in \mathcal{N}_{j,h}^{out}$. 
A key insight is that if $i$ is decided not in $\mathcal{N}_{j,h-1}^{e}$ but in $\mathcal{N}_{j,h}^{e}$, then it is very likely that $i\in\mathcal{N}_{j,h}^{out}$ is true. 
Starting from this point, we utilize a single-time excitation input and do $h$-rounds tests to achieve the inference goal. 
Two auxiliary functions are defined as  
\begin{align}
\!\! &F_0(z,e^j) \!=\! \int_{\frac{ z e^j }{2}}^{+ \infty}  \!\!\frac{ 1 }{\sqrt{2\pi}\sigma_{\omega,h}} \exp{ (-\frac{r^2 }{2\sigma^2_{\omega,h} }) } \mathrm{d}r,  \\
\!\! &F_1(z,e^j)\!=\!\int_{ \frac{ z e^j }{2} }^{ + \infty } \!\! \frac{ 1 }{\sqrt{2\pi}\sigma_{\omega,h}} \exp{ (-\frac{ (r- z e^j )^2 }{2\sigma^2_{\omega,h} }) } \mathrm{d}r, 
\end{align}
where $z\in[0,1]$. 
Based on $F_0(z)$ and $F_1(z)$, the inference probability of multi-hop out-neighbors is presented as follows.

\begin{theorem}[Lower probability bound of neighbor inference]\label{th:neighbor}
Given the maximum false alarm probability $\alpha$ of hypothesis test (\ref{eq:hypothesis2}), 
if the single-time excitation input $e^j \ge {e^j_m} =\frac{2\sqrt{2}\sigma \text{erf}^{-1} (1-2 \alpha)}{\Gamma_{ij}^{\min}}$, then we have 
\begin{align}\label{eq:h-neighbor-lowerbound}
\!\! \Pr\{i\in\mathcal{N}_{j,h}^{out}\} \!\ge\! F_1( \Gamma_{ij}^{\min} , e^j_m)(2\!-\!\alpha\!-\! F_1( \Gamma_{ij}^{\max} , e^j_m)), 
\end{align}
where $\Gamma_{ij}^{\min}$ and $\Gamma_{ij}^{\max}$ are given by 
\begin{equation}
\left\{
\begin{aligned}
&\Gamma_{ij}^{\min} = \min\{ \Gamma_{ij}(l),l=1,\cdots,h\},\\
&\Gamma_{ij}^{\max} = \max\{ \Gamma_{ij}(l),l=1,\cdots,h\}.
\end{aligned}
\right.
\end{equation}
\end{theorem} 
\begin{proof}
The proof consists of three steps. 
Denote the false alarm probability by $\delta_f(h)=\Pr\{D_1(h)|H_0(h)\}$ and the missed detection probability by $\delta_m(h)=\Pr\{D_0(h)|H_1(h)\}$.
We first prove the critical excitation magnitude for identifying the neighbors within $h$-hops. 
Then, we find the lower and upper bounds of $\delta_f(l)$ and $\delta_d(l)$.

Based on the famous Neyman-Pearson rule, with a specified $\delta_f=\alpha$, one has 
\begin{align}\label{eq:eq-alpha}
\alpha \!=\! \int_{z_0}^{+ \infty} \!\!\! \frac{ 1 }{\sqrt{2\pi}\sigma_{\omega,h}} \exp{ (-\frac{z^2 }{2\sigma^2_{\omega,h} }) } \mathrm{d}z \!=\! \frac{(1-\text{erf}(\frac{z_0}{\sqrt{2}\sigma_{\omega,h}}))}{2}. 
\end{align}
It follows from (\ref{eq:eq-alpha}) that 
\begin{align}\label{eq:error_alpha}
z_0=\sqrt{2}\sigma_{\omega,h} \text{erf}^{-1} (1-2 \alpha).
\end{align}
Due to the prior probabilities $\Pr\{H_1\}=\Pr\{H_0\}$ and based on Lemma \ref{lemma:hypothesis2}, $z_0=\frac{\Gamma_{ij}(h) e^j }{2} $ also holds at $h$-step response. 
Substituting it into (\ref{eq:error_alpha}), it yields that 
\begin{align}
e^j = \frac{2\sqrt{2}\sigma_{\omega,h} \text{erf}^{-1} (1-2 \alpha)}{\Gamma_{ij}(h)}.
\end{align}

Next, note that $F_0(z,e^j)$ decreases with $ze^j$ increasing. 
If the excitation input is designed such that 
\begin{align}
e^j_{m} = \frac{2\sqrt{2}\sigma \text{erf}^{-1} (1-2 \alpha)}{  \min\{ \Gamma_{ij}(h),h=1,\cdots,n\} },
\end{align}
then one infers that 
\begin{align}
\delta_f(l)=F_0( \Gamma_{ij}(l) , e^j_m) \le \alpha,~\forall 1\le l\le h .
\end{align}
Meanwhile, recall the detection probability $\delta_d(h)=\Pr\{D_1(h)|H_1(h)\}$ is calculated by 
\begin{align}\label{eq:detection}
\delta_d(h)\!=\!\int_{ \frac{\Gamma_{ij}(h) e^j }{2} }^{ \infty } \! \frac{ 1 }{\sqrt{2\pi}\sigma_{\omega,h}} \exp{ (-\frac{ (z- \Gamma_{ij}(h) e^j )^2 }{2\sigma^2_{\omega,h} }) } \mathrm{d}z. 
\end{align}
Since $\delta_d(h)$ increases with $\Gamma_{ij}(h) e^j_{\max}$ increasing, one has 
\begin{align}\label{eq:detection2}
F_1( \Gamma_{ij}^{\min} , e^j_m) \le \delta_d(h)\le F_1( \Gamma_{ij}^{\max} , e^j_m) .
\end{align}

Finally, utilizing the Law of Total Probability, the probability that $i$ is decided as member of $\mathcal{N}_{j,h}^{out}$ is calculated by 
\begin{small}
\begin{align}\label{eq:h-neighbor}
\!\!\!\Pr\{i\!\in\!\mathcal{N}_{j,h}^{out}\}\!=\!\Pr\{D_1(h)|H_1(h)\} \Pr\{D_0(h\!-\!1)|H_0(h\!-\!1)\} &   \nonumber  \\
+\Pr\{D_1(h)|H_1(h)\}\Pr\{D_0(h\!-\!1)|H_1(h\!-\!1)\}&.
\end{align}
\end{small}
\!\!Substitute $\Pr\{D_0(h-1)|H_0(h-1)\}=1-\delta_f(h-1)$ and $\Pr\{D_0(h-1)|H_1(h-1)\}=1-\delta_d(h-1)$ into (\ref{eq:h-neighbor}), and it yields that 
\begin{align}\label{eq:h-neighbor2}
\!\!\!\! \Pr\{i\in\mathcal{N}_{j,h}^{out}\}=&\delta_d(h)(2-\delta_f(h-1)-\delta_d(h-1))    \nonumber  \\
\ge & F_1( \Gamma_{ij}^{\min} , e^j_m)(2\!-\!\alpha\!-\! F_1( \Gamma_{ij}^{\max} , e^j_m)). 
\end{align}
The proof is completed. 
\end{proof}

Theorem \ref{th:neighbor} provides the lower probability bounds for $\Pr\{i\in\mathcal{N}_{j,h}^{out}\}$ given the maximum false alarm probability $\alpha$ of the test (\ref{eq:hypothesis2}). 
Note that the test (\ref{eq:hypothesis2}) is implemented multiple rounds to infer the neighbor within $h'$-hop, $h'=1,\cdots,h$, respectively. 
Therefore, a notable characteristic of the bounds by \eqref{eq:h-neighbor-lowerbound} is that they can be calculated recursively with just one-shot excitation input. 
The higher the hop number is, the lower the probability bound is. 
The practical application of this test is similar to (\ref{eq:1-hop-decision}) and omitted here.

\subsection{Extensions and Discussions}
Theorem \ref{th:hypothesis} and \ref{th:neighbor} illustrate the conditions and performances of using just one-time excitation. However, there are also situations where a large excitation input is not allowed in the network dynamics, making the methods not directly available. 
To overcome this deficiency, multi-excitation is a promising alternative to achieve the inference goal. 
In this part, we will briefly show how to address the issue. 

Suppose node $j$ is excited $m$ times with the same excitation input $e^j$, the inference center obtains the average observation deviation of $m$ rounds by $\tilde y^{i,\Delta}_{\bar m}= \frac{1}{m}\sum\limits_{l= 1}^{m} y^{i,\Delta}{(l)}$. 
\begin{corallary}[Upper bound of the misjudgement probability under multiple excitations]\label{th:multi-excitation}
Given excitation input $e^j>0$ and implement $m$ times of excitations, the misjudgement probability satisfies
\begin{equation}
\delta_e (m) \le 2 \int_{ \frac{ q_0 e^j}{2} }^{+ \infty} \frac{ 1 }{\sqrt{2\pi}\sigma/\sqrt{m} } \exp{ (-\frac{z^2 }{2\sigma^2 /m }) } \mathrm{d}z, 
\end{equation}
where $q_0=\min\{w_{ij}: j\in\mathcal{V}\}$. 
\end{corallary}
\begin{proof}
Based on the independent identically distributed characteristic of $y^{i,\Delta}{(l)}$, $y^{i,\Delta}{(l)}$ is subject to $N(0,\frac{\sigma^2_{\omega}}{m})$. 
Then, the misjudgement probability is calculated by 
\begin{align}
\delta_e(m)=&\Pr\{D_1|H_0\}+\Pr\{D_0|H_1\} \nonumber \\
=& \int_{ \frac{w_{ij}e^j}{2} }^{+ \infty} \frac{ 1 }{\sqrt{2\pi}\sigma_{\omega}/\sqrt{m} } \exp{ (-\frac{z^2 }{2\sigma^2_{\omega} /m }) } \mathrm{d}z \nonumber \\
& +  \int_{- \infty}^{ \frac{w_{ij}e^j}{2} } \! \frac{ 1 }{\sqrt{2\pi}\sigma_{\omega} /\sqrt{m} } \exp{ (-\frac{ (y-w_{ij}e^j)^2 }{2\sigma^2_{\omega} /m }) } \mathrm{d}z \nonumber \\
=& 2 \int_{ \frac{w_{ij}e^j}{2} }^{+ \infty} \frac{ 1 }{\sqrt{2\pi}\sigma_{\omega}/\sqrt{m} } \exp{ (-\frac{z^2 }{2\sigma^2_{\omega} /m }) } \mathrm{d}z \nonumber \\
\le &  2 \int_{ \frac{ q_0 e^j}{2} }^{+ \infty} \frac{ 1 }{\sqrt{2\pi}\sigma_{\omega}/\sqrt{m} } \exp{ (-\frac{z^2 }{2\sigma^2_{\omega} /m }) } \mathrm{d}z,
\end{align}
which completes the proof. 
\end{proof}

From Corallary \ref{th:multi-excitation}, we have that the variance $\sigma^2_{\omega} /m$ and $\delta_e(m)$ will decrease as $m$ grows. 
Therefore, it follows that 
\begin{equation}
\mathop {\lim }\limits_{m \to \infty } \delta_e(m) =0. 
\end{equation}
Corollary \ref{th:multi-excitation} illustrates that even when the magnitude of the excitation input is constrained, the misjudgment probability can be significantly reduced by increasing the excitation times. 
Due to $w_{ij}$ is not priorly known, we can relax the decision threshold as in (\ref{eq:1-hop-decision}). 
Given the maximum available excitation input $e^j_{\max}$ and specified weight threshold $\underline{w}_{ij}$, one has with probability at least $1-\bar\delta_{e,m}$ to discriminate whether $i\in\mathcal{N}_j^{out}$ by the following multiple excitation testing
\begin{equation}\label{eq:multi-time-decision}
\!\left\{
\begin{aligned}
&i\in\mathcal{N}_j^{out},&&\text{if}~|\tilde y^{i,\Delta}_{\bar m}| \ge \frac{ \sum\nolimits_{l= 1}^{m} \Delta y^{\max}(l)  }{m} + \frac{\underline{w}_{ij}e^j_{\max}}{2}, \\
&i\notin\mathcal{N}_j^{out},&&\text{else},
\end{aligned}
\right.
\end{equation}
where $\bar\delta_{e,m}=2 \int_{ \frac{ \underline{w}_{ij} e^j_{\max} }{2} }^{+ \infty} \frac{ 1 }{\sqrt{2\pi}\sigma_{\omega}/\sqrt{m} } \exp{ (-\frac{z^2 }{2\sigma^2_{\omega} /m }) } \mathrm{d}z$. 
A minor drawback of this method is that if the weight between two nodes is small and the excitation time is also limited, an existing edge may be regarded as not existing.

\begin{figure*}[t]
\centering
\subfigure[One-hop neighbor inference of a node using different excitation inputs.]{\label{excitation_asymptotic}
\includegraphics[width=0.32\textwidth]{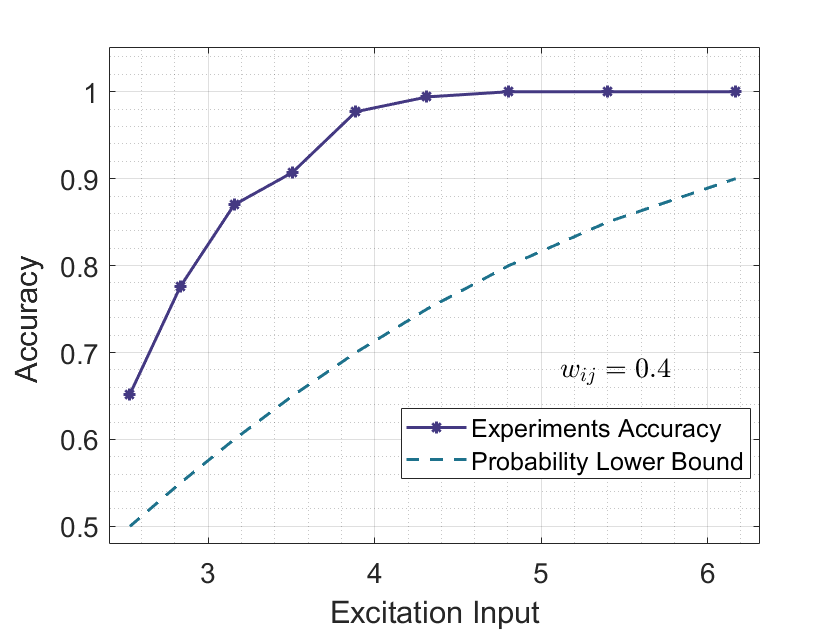}}
\subfigure[Multi-hop neighbor inference of a node using the same excitation input.]{\label{excitation_multi}
\includegraphics[width=0.32\textwidth]{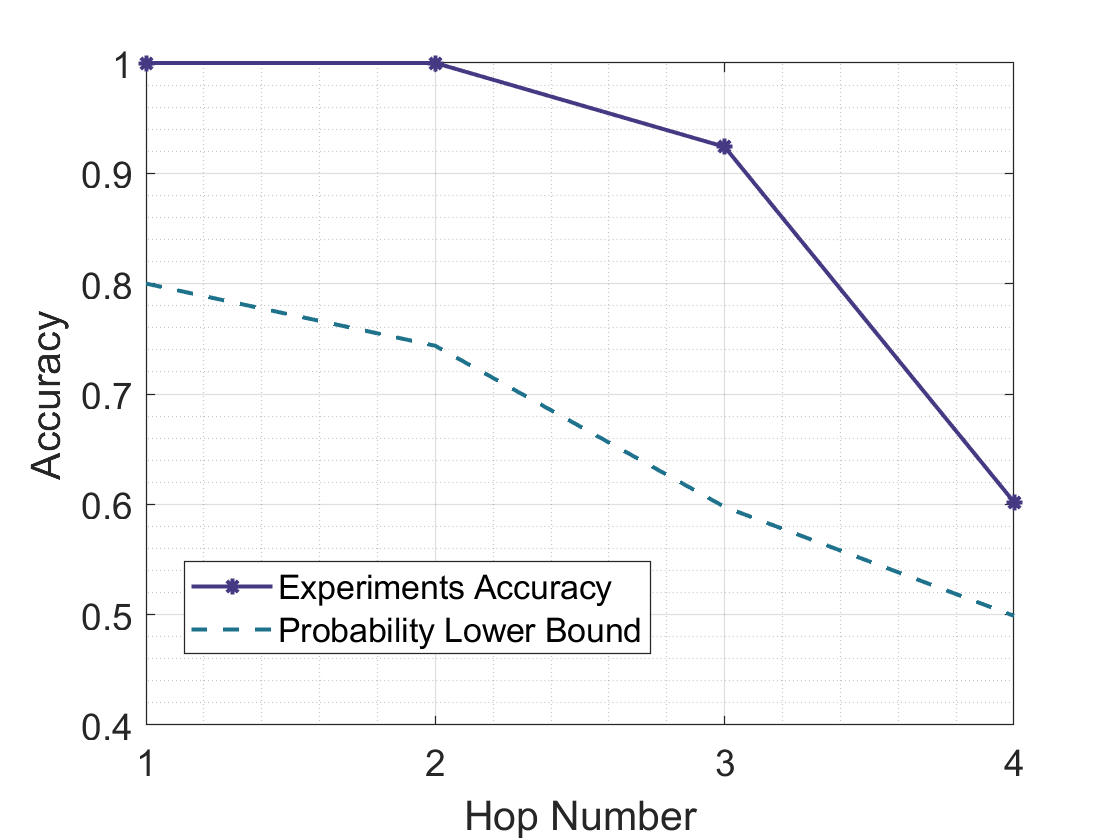}}
\subfigure[Comparisons of the topology inference without and with the excitation based method.]{\label{fig:improve}
\includegraphics[width=0.32\textwidth]{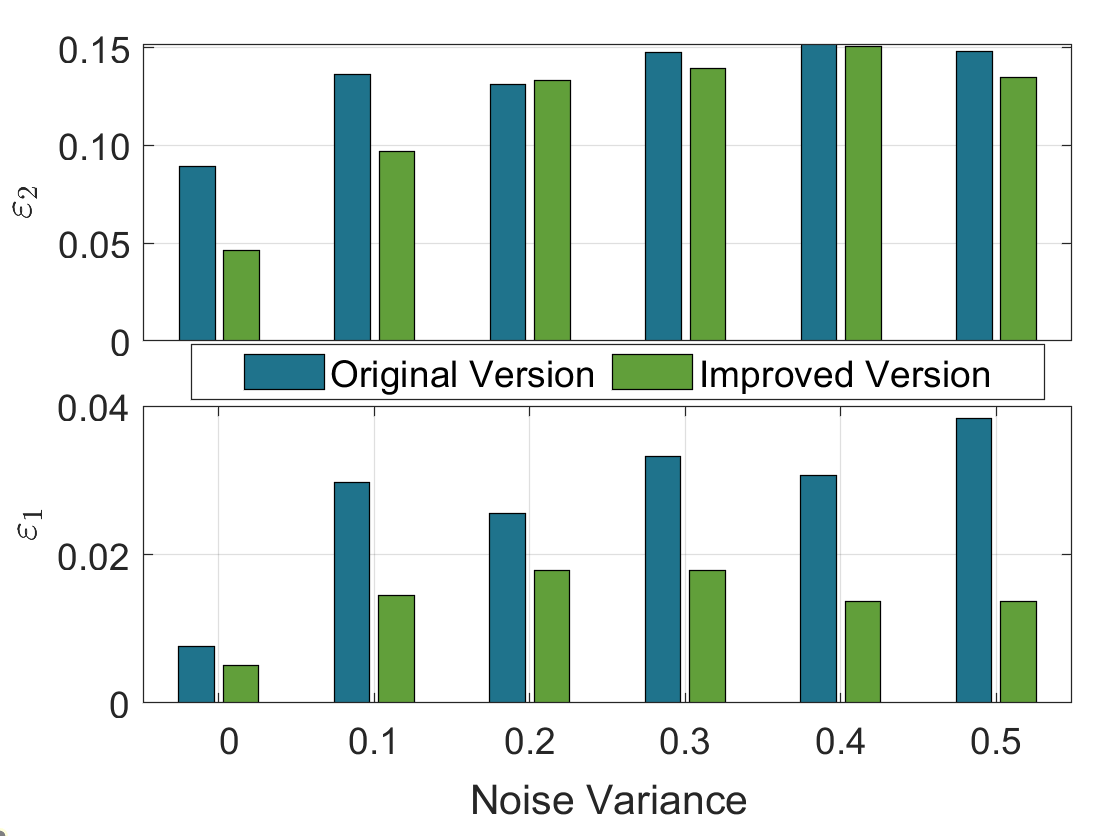}}
\caption{Simulation results of the proposed excitation-based method. 
The experiment accuracy in (a)-(b) is obtained by implementing the hypothesis test $1000$ times and then computing the ratio of positive results. 
}
\label{final_excitation}
\vspace*{-10pt}
\end{figure*}

Finally, we illustrate how to use the excitation method to prove the performance of existing inferences. 
Suppose the observer has gained the observations from $0$ to $t$ moments. 
Traditionally, the inference problem can be formulated as solving the ordinary least square problem 
\begin{equation}\label{eq:olsprob}
\hat{W}=\mathop {\arg }\limits_{W} \min \sum\nolimits_{m = 1}^{t+1} \| y_t - W y_{t-1}  \|^2_2. 
\end{equation}
By the excitation-based method, we can inject the excitation input $e^j$ on node $j$ at moment $t$. 
Based on the observation $\tilde{y}_{t+1}$, the results of the excitation based inference method are utilized to solve the following constrained least square problem
\begin{subequations} \label{eq:improve}
\begin{align}
\mathop {\min }\limits_{W} ~& \sum\nolimits_{m = 1}^{t} \| y_t - W y_{t-1}  \|^2_2  \\
\text{s.t.}~~&W_{ij}>0,~\text{if}~ |\tilde y^{i,\Delta}_{t+1}| \ge \Delta y_{t}^{\max} + \frac{\underline{w}_{ij} |e^j|}{2}, \label{eq:improve_b}\\
&W_{ij}=0,~\text{if}~ |\tilde y^{i,\Delta}_{t+1}| < \Delta y_{t}^{\max} + \frac{\underline{w}_{ij} |e^j|}{2}. \label{eq:improve_c}
\end{align}
\end{subequations}
By solving problem \eqref{eq:improve}, the final inferred global topology has smaller errors compared with that of \eqref{eq:olsprob}. 

\begin{remark}
The key insight of improving the inference accuracy of \eqref{eq:olsprob} lies in that the topology is estimated in the independently row-by-row manner (i.e., solving $W_{[i,:]}$). 
Since the connections between $j$ and $\mathcal{N}_{j}^{out}$ constitute a column of $W$, the explicit constraints \eqref{eq:improve_b} and \eqref{eq:improve_c} for $w_{ij}$ reduce the uncertainty of all other elements in $i$-th row of $W$, thus making the global inference accuracy improved. 
\end{remark}

\section{Numerical Simulations}\label{sec:simulation}
In this section, we present numerical simulations to demonstrate the performance of the analytical results. 
First, we display the basic setup. 
Then, we conduct groups of experiments under different conditions, including the system stability and noise variance. 
Detailed analysis is also provided to demonstrate the performance of the proposed method. 

The most critical components are the adjacent matrix $A$ and the interaction matrix $W$.
For the setting of interaction matrix $W$, we randomly generate a directed topology structure with $|\mathcal{V}|=20$, and the weight of $W$ is designed by the Laplacian rule. 
To save space, we mainly present the results of the case $W\in \mathcal{S}_m$ (the results of case $W\in \mathcal{S}_a$ are likewise). 
For generality, the initial states of all agents are randomly selected from the interval $[-100,100]$, and the variance of the process and observation noise satisfy $\sigma_{\theta}^2=1$ and $\sigma_{\upsilon}^2=1$.

Now, we move on to verify the performance excitation-based method, as shown in Fig.~\ref{final_excitation}. 
First, we excite a target node $j$ and wish to find its one-hop out-neighbor $i$ subject to $w_{ij}\ge0.4$. 
Given the lower probability bound $\bar \delta_e$, the critical excitation input is calculated by $|e^j|=\frac{2\sqrt{2} \sigma_{\omega({i}) }}{w_{ij}} \text{erf}^{-1}(1-\bar \delta_e)$. 
We use the input to conduct the hypothesis test $1000$ times and compute the ratio of positive results. 
As one expects, considering the same one-hop neighbor connection to be inferred, larger excitation input ensures higher accuracy of the decision results, as shown in Fig.~\ref{excitation_asymptotic}. 
Next, the multi-hop neighbor inference results are provided in Fig.~\ref{excitation_multi}. 
It is easy to see that given the maximum false alarm probability $\alpha$ and under the same excitation input, the accuracy for multi-hop neighbor inference will decrease as the hop number grows, which corresponds to the common intuition. 
The probability lower bound here is computed by \eqref{eq:h-neighbor-lowerbound}. 
We note that the dashed lines in Fig.~\ref{excitation_asymptotic} and Fig.~\ref{excitation_multi} are lower bounds of the accuracy in theory. 
Thus it makes sense that the actual accuracy in experiments is higher than that bound. 
The multiple excitation cases are likewise and are omitted here. 

Finally, we provide the results of improving the inference performance of the causality based estimator in \cite{lys-cdc2021} to solve \eqref{eq:olsprob}. 
Here we directly present the case $W\in \mathcal{S}_m$ and consider the following two error indexes
\begin{align}
{\varepsilon _1} &=({{\| {\text{sign}({\hat W}) - \text{sign}(W) \|}_0}})/n^2, \\
{\varepsilon _2} &= ({\| {{\hat W} - W} \|_{\text{Frob}} }) / {\| W\|_{\text{Frob}}},
\end{align}
which represents the structure and magnitude errors, respectively. 
As we can see from Fig.~\ref{fig:improve}, with the same observations, the inference error is largely reduced by combining the excitation based inference results to solve \eqref{eq:improve}, especially in terms of the structure error.

\section{Conclusions}\label{sec:conclusion}
In this paper, we investigated the topology inference problem of NDSs by using very few excitations. 
First, we introduced the definition of $h$-hop neighbor and proposed the one-shot excitation based method. 
By utilizing the tool of hypothesis testing, we proved the magnitude condition of the excitation input with probability guarantees. 
Then, we extended the one-hop inference method to $h$-hop neighbor and multiple excitations cases. 
The inference accuracy was rigorously analyzed. 
Finally, the performance study by simulations verified our performance analysis. 
The proposed inference method is helpful in scenarios of insufficient observations over NDSs, and can also be used as auxiliary means to improve the accuracy of existing methods. 
Future directions include extending the method to infer the specific values of the global topology, and making the few excitations cooperate to finish the inference task.




\balance


\end{document}